\documentclass[10pt,a4paper]{amsart}
\usepackage[utf8]{inputenc}
\usepackage{amsmath}
\usepackage{amsfonts}
\usepackage{amssymb}
\usepackage{amsthm}
\usepackage{graphicx}
\usepackage{dsfont}
\usepackage{stmaryrd}
\usepackage{xcolor}
\usepackage{tikz}
\usepackage{subcaption}
\usetikzlibrary{patterns}
\usetikzlibrary{backgrounds,scopes}
\usetikzlibrary{calc}
\usetikzlibrary{intersections}
\usepackage{url}

\usepackage{hyperref}

\usepackage[all]{xy}

\newcommand{\beq}{\begin{equation}}
\newcommand{\eeq}{\end{equation}}

\makeatletter
\newcommand{\vast}{\bBigg@{4}}
\newcommand{\Vast}{\bBigg@{5}}
\makeatother

\newtheorem{theorem}{Theorem}[section]
\newtheorem{proposition}[theorem]{Proposition}
\newtheorem{corollary}[theorem]{Corollary}
\newtheorem{lemma}[theorem]{Lemma}

\theoremstyle{definition}
\newtheorem{definition}[theorem]{Definition}

\theoremstyle{remark}
\newtheorem{remark}[theorem]{Remark}

\newcommand{\ep}{\epsilon}

\renewcommand{\appendix}[1]{
    \addtocounter{section}{1}
    \setcounter{equation}{0}
    \renewcommand{\thesection}{\Alph{section}}
    \section*{Appendix \thesection\protect\indent #1}
    \addcontentsline{toc}{section}{Appendix \thesection\ \ \ #1}
}
\newcommand\encadremath[1]{\vbox{\hrule\hbox{\vrule\kern8pt
\vbox{\kern8pt \hbox{$\displaystyle #1$}\kern8pt}
\kern8pt\vrule}\hrule}}
\def\enca#1{\vbox{\hrule\hbox{
\vrule\kern8pt\vbox{\kern8pt \hbox{$\displaystyle #1$}
\kern8pt} \kern8pt\vrule}\hrule}}

\newcommand\figureframex[3]{
\begin{figure}[bth]
\hrule\hbox{\vrule\kern8pt
\vbox{\kern8pt \vbox{
\begin{center}
{\mbox{\epsfxsize=#1.truecm\epsfbox{#2}}}
\end{center}
\caption{#3}
}\kern8pt}
\kern8pt\vrule}\hrule
\end{figure}
}
\newcommand\figureframey[3]{
\begin{figure}[bth]
\hrule\hbox{\vrule\kern8pt
\vbox{\kern8pt \vbox{
\begin{center}
{\mbox{\epsfysize=#1.truecm\epsfbox{#2}}}
\end{center}
\caption{#3}
}\kern8pt}
\kern8pt\vrule}\hrule
\end{figure}
}

\newcommand{\bea}{\begin{eqnarray}}
\newcommand{\eea}{\end{eqnarray}}

\renewcommand{\and}{{\qquad {\rm and} \qquad}}

 \newcommand{\Tr}{{\,\rm Tr}\:}
\newcommand{\tr}{{\,\rm tr}\:}

\newcommand{\Pint}{{\int\kern -1.em -\kern-.25em}}

\renewcommand{\thesection}{\arabic{section}}

\makeatletter
\@addtoreset{equation}{section}
\makeatother
\newcommand{\brem}{\begin{remark}\rm\small}
\newcommand{\er}{\end{remark}}
\newcommand{\bt}{\begin{theorem}}
\newcommand{\et}{\end{theorem}}
\newcommand{\bd}{\begin{definition}}
\newcommand{\ed}{\end{definition}}
\newcommand{\bp}{\begin{proposition}}
\renewcommand{\ep}{\end{proposition}}
\newcommand{\bl}{\begin{lemma}}
\newcommand{\el}{\end{lemma}}
\newcommand{\bc}{\begin{corollary}}
\newcommand{\ec}{\end{corollary}}
\newcommand{\beaq}{\begin{eqnarray}}
\newcommand{\eeaq}{\end{eqnarray}}

\newcommand\R{\mathcal{R}}
\def\boxSize{20}
\def\boxPic{
  \put(0,0){
    \put(0,0){\line(0,-1){\boxSize}}
    \put(0,0){\line(1,0){\boxSize}}
    \put(\boxSize,0){\line(0,-1){\boxSize}}
    \put(0,-\boxSize){\line(1,0){\boxSize}}
  }
}
\def\diagBoxPic{
  \put(0,0){
    \put(0,10){\line(1,-1){10}}
    \put(0,10){\line(-1,-1){10}}
    \put(0,-10){\line(1,1){10}}
    \put(0,-10){\line(-1,1){10}}
  }
}

\begin{document}

\vspace{-12cm}

\begin{center}
  \hfill MIPT/TH-26/24\\
  \hfill ITEP/TH-33/24\\
  \hfill IITP/TH-28/24
\end{center}


\title{Direct proof of one-hook scaling property for Alexander polynomial \\
  from Reshetikhin-Turaev formalism}

\author[And.~Morozov]{And.~Morozov}
\address{And.~M.: NRC ``Kurchatov Institute'', IITP, MIPT}
\email{andrey.morozov@itep.ru}

\author[A.~Popolitov]{A.~Popolitov}
\address{A.~P.: MIPT, IITP, NRC ``Kurchatov Institute''}
\email{popolit@gmail.com}

\author[A.~Sleptsov]{A.~Sleptsov}
\address{A.~S.: MIPT, IITP, NRC ``Kurchatov Institute''}
\email{sleptsov@itep.ru}

\begin{abstract}
  We prove that normalized colored Alexander polynomial (the $A\rightarrow 1$ limit of colored HOMFLY-PT polynomial)
  evaluated for one-hook (L-shape) representation $R$ possesses
  \textit{scaling property}: it is equal to the fundamental Alexander polynomial
  with the substitution $q \rightarrow q^{|R|}$. The proof is simple and direct use of Reshetikhin-Turaev formalism
  to get all required $\mathcal{R}$-matrices.
\end{abstract}

\maketitle

{\section{Introduction}\label{sec:introduction}
  Today, knot theory is very active and flourishing branch of mathematics, having immediate connections
  not only to other branches of mathematics, such as quantum group representation theory
  \cite{book:KC-quantum-calculus,
    book:C-group-theory-birdtracks,
    book:K-quantum-groups,
    book:KS-quantum-groups-and-their-representations,
    paper:C-jones-polynomials-of-torus-knots-via-daha}
  and the theory of quantum computation \cite{paper:K-graphical-calculus-canonical-bases-and-kl-theory,
    book:B-quantum-computing-for-everyone,
    book:Y-an-introduction-to-quantum-computing},
  but also to theoretical physics,
  where its functional integral counterpart -- the Chern-Simons theory
  \cite{paper:CS-characteristic-forms-and-geometric-invariants}
  -- is currently at the border between known and unknown:
  the more-or-less completely well-understood 2D conformal field theories
  \cite{book:FMS-conformal-field-theory}
  and still relatively mysterious real-world 4D quantum field theories
  \cite{paper:S-persistent-challenges-of-quantum-chromodynamics}.
  
  Central objects of knot theory are the so-called polynomial knot invariants: multivariate Laurent polynomials,
  algorithmically constructed from a knot and invariant under its ambient isotopy
  \cite{paper:A-the-knot-book}.
  Hence, differing knot polynomials mean that two knots are inequivalent.
  
  Perhaps, the most widely known knot polynomial is HOMFLY-PT polynomial
  \cite{paper:HOMFLY-a-new-polynomial-invariant,
    paper:PT-invariants-of-links-of-conway-type}.
  In its simplest, fundamental, version it can be calculated for any knot using nothing but a simple skein relation
  and Reidemeister moves. However, fundamental HOMFLY-PT is not the ``complete'' knot invariant, that is,
  not all knot classes are distinguished by it. For instance, any pair of so-called mutant knots has same HOMFLY-PT
  \cite{paper:BDGMMMRSS-distinguishing-mutant-knots}.

  Of real interest are, hence, \textit{colored} HOMFLY-PT polynomials
  \cite{paper:B-overview-of-knot-invariants-at-roots-of-unity,
    paper:BMM-invariants-of-knots-and-links-at-roots-of-unity}.
  Not only do they start to distinguish
  mutants, they also are the values of Wilson loop averages in Chern-Simons theory
  \cite{paper:Witten}.  The monodromy trace of a Wilson loop should be taken in some representation,   and this corresponds to the color of the HOMFLY-PT.

  In principle, any concrete colored HOMFLY-PT can be calculated using the Reshetikhin-Turaev
  formalism \cite{paper:MS-cs-theory-in-the-temporal-gauge}:
  one just needs to take the quantum trace of the product of $R$-matrices in suitable representations.
  Here, however, lies the difficulty: the necessary colored $R$-matrices or, equivalently, the Racah-Wigner matrices 
  (assembled from Racah W-coefficients \cite{paper:R-theory-of-complex-spectra-ii},
  also known in different context under the name of Wigner's 6j-symbols
  \cite{paper:W-on-the-matrices-which-reduce})
  that diagonalize them, are not known in general and their straightforward calculation, so to say, from the definition,
  for instance, with help of the highest weight method, is computationally complex
  \cite{    paper:MMMS-quantum-racah-matrices-and-3-strand-braids-in-irreps-with4,
    paper:BJLMMMS-quantum-racah-matrices-upto-level-3}.
  In practice this means that, even with modern computers, colored HOMFLY-PT can be calculated
  for generic knot only for representations $\lambda$ with small number of boxes
  (see \cite{site:knotebook} for current progress).
  
  At the same time, (quantum) Racah-Wigner matrix, being a representation theoretic object,
  is needed also outside of knot theory context. For instance, it is a building block
  for universal gates of topological quantum computer
  \cite{paper:KM-quantum-r-matrices-as-universal-qubit-gates}. Moreover, in the TQFT approach
  to invariants of 3D manifolds, it is 6j-symbol that is associated
  to the ideal hyperbolic tetrahedron -- the building block of 3D hyperbolic geometry
  \cite{paper:turaev1992state}.
  
  Therefore, some indirect way to calculate Racah-Wigner matrix,
  based on some emergent properties,
  or hidden symmetries, is highly desirable. One way to discover such symmetries is to look
  for the symmetries of the polynomial knot invariants themselves
  \cite{paper:MST-a-novel-symmetry-of-colored-homfly,
    paper:MST-a-new-symmetry-of-colored-alexander}.

  And indeed, this line of thought had been, at least partially, successful. One can calculate 6j-symbols from a small class of polynomial knot invariants \cite{paper:MFS, paper:MMS-colored-homfly-polynomials-for-the-pretzel-knots-and-links,
    paper:GMMMS-colored-knot-polynomials-for-arbitrary-pretzel-knots-and-links} and then using effective field-theory descriptions extend calculation of colored HOMFLY-PT polynomials to a huge class of knots \cite{paper:MMMRS-colored-homfly-polynomials-for-knots-presented-as-double-fat-diagrams, paper:MMMRSS-tabulating-knot-polynomials-for-arborescent-knots}.

One more powerful approach is due to the celebrated eigenvalue conjecture
  \cite{    paper:IMMM-eigenvalue-hypothesis-for-racah-matrices, paper:DMMMRSS-eigenvalue-hypothesis-for-multistrand-braids, paper:DMMMRSS-multicolored-links-from-3-strand-braids}: the statement that any $R$-matrix (and, hence, corresponding Racah-Wigner matrices)
  block is, in fact, uniquely determined by the set of its eigenvalues.
  All of these properties have allowed to push boundaries of possible colored HOMFLY-PT calculations, but there is still much more to be done.

  Some of the emergent properties are first observed in a simplified version of HOMFLY-PT
  polynomial -- the Alexander polynomial. From HOMFLY-PT perspective, it is just
  the $A \rightarrow 1$ limit, but historically Alexander polynomial appeared first
  \cite{book:C-knots-and-links},
  and it has an independent definition in terms of the fundamental group
  of the torus vicinity of a knot.
  The properties, that have been observed recently, include block symmetry
  \cite{paper:MST-a-novel-symmetry-of-colored-homfly,
    paper:MST-a-new-symmetry-of-colored-alexander}
  and the \textit{1-hook scaling property}
  \cite{paper:IMMM-homfly-and-superpolynomials-for-figure-eight-knot-in-all-symmetric-and-anti-reps,
    paper:MM-eigenvalue-conjecture-and-colored-alexander-polynomials},
  which is the main focus of the present paper.

  The 1-hook scaling property is a statement that colored Alexander polynomial, evaluated for an L-shape (that is, one-hook)
  representation $\lambda$ is related to the fundamental Alexander polynomial by simple substitution
  \begin{theorem} \label{thm:alexander-trivialization}
    \begin{align}
      \mathcal{A}_\lambda(q) = \lim_{A \rightarrow 1} H(A,q) = \mathcal{A}_\Box (q^{|\lambda|}),
    \end{align}
    where $\lambda$ is a one-hook partition \eqref{eq:one-hook-partition}
    and HOMFLY-PT polynomial is defined with help of Reshetikhin-Turaev formula \eqref{eq:homfly-definition}
  \end{theorem}

  \begin{remark}
    This theorem~\ref{thm:alexander-trivialization} is non-trivial \textit{despite}
    widespread point-of-view that it is a simple corollary of
     factorization property for Alexander polynomials for satellite knots:
    \begin{align} \label{eq:satellite-factorization}
      \Delta_S(q) = \Delta_P(q) \cdot \Delta_C(q^w),
    \end{align}
    where $\Delta$ is the fundamental Alexander polynomial,
    $S$ is a satellite knot with ``pattern'' knot $P$ and ``companion'' knot $C$,
    and $w$ is the number of Dehn twists (windings) of the corresponding solid torus.
    This widespread point-of-view is, however, \textit{wrong}. Indeed, taken at face value,
    this reasoning should equally imply scaling property for arbitrary representations
    and not just 1-hook, which is straightforwardly verified to not be the case
    from the examples available
    \cite{paper:BDGMMMRSS-distinguishing-mutant-knots,
      paper:BJLMMMS-quantum-racah-matrices-upto-level-3,
      paper:SS-quantum-racah-matrices-and-3-strand-braids}.
    Second, the more accurate application of \eqref{eq:satellite-factorization}
    leads to the relation between peculiar \textit{sums} of colored Alexander polynomials
    for different representations $\lambda$, and not just isolated polynomials.
  \end{remark}
  
  Experimentally, the statement of Thm.~\ref{thm:alexander-trivialization}
  was known for rather long time,
  moreover, in \cite{paper:MM-eigenvalue-conjecture-and-colored-alexander-polynomials}
  it was proved to be true, provided eigenvalue conjecture holds;
  and in \cite{paper:Z-colored-homfly-via-skein-theory} it was proved in the case of torus knots $T[m,n]$.
  Since 1-hook scaling property is very easy to observe experimentally, this provided additional evidence in favor of the eigenvalue conjecture.
  Furthermore, in papers \cite{paper:MMMMS-coloured-alexander-polynomials-and-kp-hierarchy,
    paper:MS-perturbative-analysis-of-the-colored-alexander-polynomial}
  the connection was discovered between
  colored 1-hook Alexander polynomials and one-soliton solutions of KP hierachy;
  where it was derived provided the 1-hook scaling property
  Th.~\ref{thm:alexander-trivialization} holds.
  In this paper we prove the 1-hook scaling property (Theorem~\ref{thm:alexander-trivialization}) from first principles, that is, directly applying the Reshetikhin-Turaev formalism,
  thus naturally completing the cycle of works
  \cite{paper:MMMMS-coloured-alexander-polynomials-and-kp-hierarchy,
    paper:MS-perturbative-analysis-of-the-colored-alexander-polynomial}.

  The crucial element of the proof is a clever application of the Yang-Baxter equation
  \cite{paper:T-the-yang-baxter-equation-and-invariants-of-links},
  of which we actually use only a small part. Nevertheless, it is sufficient to fix all the $R$-matrix blocks.
  We hope that the presented technique will be useful in more general calculations
  of Racah-Wigner matrix blocks.
  
  \bigskip

  Furthermore, while still being in the conjectural state, the 1-hook scaling property
  for Alexander polynomials inspired big research direction: the search for similar
  scaling symmetries for colored HOMFLY polynomials
  \cite{paper:MST-a-novel-symmetry-of-colored-homfly,
    paper:LST-chern-simons-perturbative-series-revisited,
    paper:LST-implications-for-colored-homfly,
    paper:LS-tug-the-hook-symmetry-for-quantum-6j-symbols},
  and for colored Alexander polynomials
  in more complicated representations
  \cite{paper:MST-a-new-symmetry-of-colored-alexander}. These recent and intriguing
  developments are currently still waiting to be proved using representation theory
  of quantum (super) groups.
  
  \bigskip
  
  This paper is organized as follows. In Section~\ref{sec:background} we review the necessary
  background on HOMFLY-PT and Alexander polynomials and the approach to their calculation using Young graph.
  In Section~\ref{sec:r-matrix-defs} we describe different $R$-matrices
  focusing on the ones we need for the present paper.
  Sections~\ref{sec:multiplication-of-two-1-hooks} and~\ref{sec:square-of-1-hook}
  are somewhat auxiliary and we build up necessary facts about the algebra of one-hook representations
  and how $R$-matrix acts in their products.
  In Section~\ref{sec:block-structure} we recall how block structure of $R$-matrices is determined
  from the Young graph. In Section~\ref{sec:one-hook-young-graph} we discuss the structure of one-hook
  sector of the Young graph and the precise form of $1 \times 1$ blocks of $R$-matrix.
  In Section~\ref{sec:first-two-by-two-block} we obtain formula for the first two by two block
  and in Section~\ref{sec:two-recursions-for-doublets} we establish two recursions
  on the Young graph, which allow us to find the form of generic two by two block on our Young graph of interest.
  In Section~\ref{sec:topological-schurs} we show that the ratio of Schur polynomials
  at topological locus in Reshetikhin-Turaev formula also simplifies in the limit $A \rightarrow 1$.
  The combination of lemmas from Sections~\ref{sec:background}-\ref{sec:topological-schurs}
  proves Theorem~\ref{thm:alexander-trivialization}.
}

{\section*{Acknowledgments}
  This work is supported by the Russian Science Foundation (Grant No.24-71-10058).
}

{\section{Background}\label{sec:background}
  In this section we describe the necessary background,
  the Reshetikhin-Turaev formula and the approach to calculation with help of Young graph
  \cite{paper:R-invariants-of-tangles-I,
    paper:AM-cabling-procedure,
    paper:AMMM-colored-homfly-polynomials-as-multiple-sums}.

  Given a knot $\mathcal{K}$ in a braid representation with $m$ strands,
  suppose that multiplicities of crossings in a braid are given by
  \begin{align}
    (a_{1,1}, a_{1,2}, \dots , a_{1,m-1}, a_{2,1}, \dots a_{2,m-1} \dots a_{n,m-1}),
  \end{align}
  that is,
  there are $a_{1,1}$ crossings between first and second strands of the braid, followed by $a_{1,2}$
  crossings between second and third and so on up to $a_{1,m-1}$ crossings between $m-1$th and $m$th strands,
  then followed again by $a_{2,n}$ crossings between first and second strand and so on.
  The negative number $a_{i,j}$ corresponds to $|a_{i,j}|$ inverse crossings.

  Then the normalized (or, in other words, reduced) HOMFLY-PT polynomial $H_\lambda(A,q)$,
  colored with some representation $\lambda$, can be
  calculated as a quantum trace of the product of corresponding $R$-matrices
  (see Section~\ref{sec:r-matrix-defs} for definitions of different $R$-matrices).
  We take this as a
  \begin{definition}
    \begin{align} \label{eq:homfly-definition}
      H_\lambda (A, q) = \sum_{\mu \vdash m|\lambda|} \frac{d_\mu}{d_\lambda}
    \Tr \prod_{i=1}^n \prod_{j=1}^{m-1} \left(\R^{(i)}_{\lambda,\mu} \right)^{a_{i,j}}.
    \end{align}
  \end{definition}
  The colored Alexander polynomial $\mathcal{A}_\lambda(q)$ is the $A \rightarrow 1$ limit of the corresponding
  HOMFLY-PT polynomial.
  \begin{definition}
    \begin{align} \label{eq:alexander-definition}
      \mathcal{A}_\lambda(q) = \lim_{A \rightarrow 1} H_\lambda(A,q)
    \end{align}
  \end{definition}
  
  Other definitions, consistent with this one, as possible. For instance, one can
  define fundamental HOMFLY-PT polynomial with help of skein relation
  \cite{paper:HOMFLY-a-new-polynomial-invariant}
  and then define colored polynomials through the cabling procedure
  \cite{paper:AM-cabling-procedure,
    paper:AMMM-colored-homfly-polynomials-as-multiple-sums,
    paper:DMMMRSS-eigenvalue-hypothesis-for-multistrand-braids}.

  The quantum dimension of representation $d_\mu$ (and also $d_\lambda$)
  is equal to the value of Schur polynomial at the so-called \textit{topological locus}
  \begin{align} \label{eq:quantum-dimension-is-schur}
    d_\mu = s_\mu \left (
    p_k = \frac{(A^k - A^{-k})}{(q^k - q^{-k})}
    \right ) =: s_\mu^*
  \end{align}
  and the quantities $\R_{\lambda,\mu}^{(i)}$ are $R$-matrices in the space of intertwiners,
  which we interchangeably call the space of multiplicities.

  The formula \eqref{eq:homfly-definition} is the essence of Reshetikhin-Turaev approach
  to HOMFLY-PT polynomial and to knot invariants in general
  \cite{paper:MS-cs-theory-in-the-temporal-gauge}.
  The sum in \eqref{eq:homfly-definition}
  is over all irreducible summands in $m$th tensor power of representation $\lambda$.

  In order to apply \eqref{eq:homfly-definition} one needs a way to
  calculate the $R$-matrices $\R_{\lambda,\mu}^{(i)}$ in the multiplicity (intertwiner) space.
  One of the bases in this space is manifestly given by comb diagrams
  (see Section~\ref{sec:block-structure} and also \cite{paper:R-invariants-of-tangles-I}),
  which can be conveniently depicted as paths on the so-called \textit{Young graph} for representation $\lambda$.
  More precisely, basis in the multiplicity space $\lambda,\mu$ is given by all paths joining
  the root $\lambda$ of the Young graph and $\mu$.
  The Young graph is a directed acyclic graph, consisting of layers. Vertices on each layer $i$ are
  the irreducible representations in $i$th tensor power of $\lambda$.
  The edge (possibly, with multiplicity) between some vertex $P$ on layer $i$ and some vertex $Q$ on layer $i+1$ occurs
  when $Q$ belongs to the tensor product of $P$ with $\lambda$.

  For instance, in case of fundamental representation $\lambda = \Box$
  the Young graph takes the following form
  
  
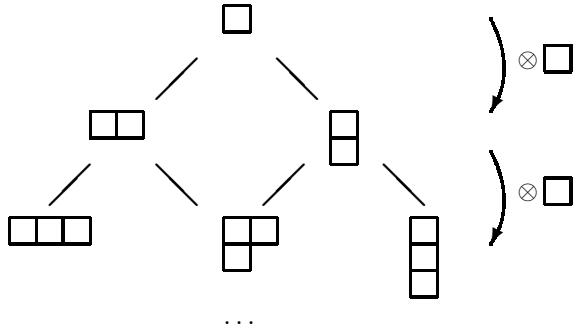
\begin{align}
    \begin{picture}(300,130)(0,-130)
      \thicklines
      \def\boxSize{10}
      \put(150,0){\boxPic}
      \put(140,-20){\line(-1,-1){15}}
      \put(170,-20){\line(1,-1){15}}
      \put(250,-5){\qbezier(0,0)(10,-20)(0,-35)}
      \put(250,-40){\vector(-1,-2){0}}
      \put(270,-15){\put(-10,-8){$\otimes$} \boxPic}
      \put(100,-40){
        \put(0,0){\boxPic}
        \put(10,0){\boxPic}
      }
      \put(190,-40){
        \put(0,0){\boxPic}
        \put(0,-10){\boxPic}
      }
      \put(100,-60){\line(-1,-1){15}}
      \put(125,-60){\line(1,-1){15}}
      \put(180,-60){\line(-1,-1){15}}
      \put(210,-60){\line(1,-1){15}}
      \put(0,-50){
        \put(250,-5){\qbezier(0,0)(10,-20)(0,-35)}
        \put(250,-40){\vector(-1,-2){0}}
        \put(270,-15){\put(-10,-8){$\otimes$} \boxPic}
      }
      \put(70,-80){
        \put(0,0){\boxPic}
        \put(10,0){\boxPic}
        \put(20,0){\boxPic}
      }
      \put(150,-80){
        \put(0,0){\boxPic}
        \put(10,0){\boxPic}
        \put(0,-10){\boxPic}
      }
      \put(220,-80){
        \put(0,0){\boxPic}
        \put(0,-10){\boxPic}
        \put(0,-20){\boxPic}
      }
      \put(150,-120){$\dots$}
    \end{picture}
  \end{align}
  and this depicted part of the Young graph is sufficient to calculate fundamental HOMFLY-PT polynomial
  for any three-strand braid.
  
  In the fundamental case, one can prove that there is an explicit combinatorial
  formula for $\R$-matrices in terms of the diagrams on the Young graph
  \cite{paper:R-invariants-of-tangles-I,
    paper:AM-cabling-procedure,
    paper:AMMM-colored-homfly-polynomials-as-multiple-sums}.
  For the generic case, however, the calculation that leads from the general definition to the combinatorial
  formula is much more difficult to perform, see
  \cite{paper:DMMMRSS-eigenvalue-hypothesis-for-multistrand-braids}
  for current state of the art.

  However, even in generic case it is very easy to establish the block structure of the $\R$-matrices
  from the Young graph, see Section~\ref{sec:block-structure} for details.

  We \textbf{define} colored Alexander polynomial to be the $A \rightarrow 1$ limit of the HOMFLY-PT polynomial.
  In case $\lambda = \Box$ this definition coincides with the historical independent definition of the Alexander polynomial
  (see, for instance \cite{book:C-knots-and-links} and references therein).

  There is an explicit formula for the quantum dimension \eqref{eq:quantum-dimension-is-schur}
  as a product over boxes of the Young diagram \eqref{eq:topological-locus-schur}.
  It is easy to see (Section~\ref{sec:topological-schurs}) that only the diagrams $\mu$ whose
  main diagonal is equal to the main diagonal of $\lambda$ contribute
  to the Alexander polynomial. Otherwise the order of the zero in the numerator is greater than in the denominator,
  and the corresponding summand vanishes.

  Therefore, for the case we are interested in, namely $\lambda$ being one-hook (L-shape) diagram,
  the diagrams $\mu$ that contribute to the answer are also one-hook.

  The structure of the Littlewood-Richarsdon rule, that governs tensor products
  of $sl(n)$ representations is such, that if a representation $\mu$ occurs in a
  product of representations $\lambda$ and $\nu$, then Young diagrams of both $\lambda$ and $\nu$
  should be subsets of $\mu$.
  This means that if we are interested in one-hook representations $\mu$ then all the representations,
  belonging to the path leading from $\lambda$ to $\mu$ should be also one-hook.

  In other words, we have a
  \begin{proposition} \label{prop:one-hook-sufficient}
    In order to calculate colored Alexander polynomial
    for the one-hook representation $\lambda$ it is sufficient to consider only the one-hook subgraph
    of the $\lambda$'s Young graph, i.e. the subgraph that consists only of one-hook Young diagrams.
  \end{proposition}
  We analyze the structure of this subgraph in detail in Section~\ref{sec:one-hook-young-graph}.
  
}

{\section{The different R-matrices}\label{sec:r-matrix-defs}
  The notion of R-matrix is highly overloaded in the literature, moreover, even in the present paper
  we need several related, but distinct, operators each called an R-matrix in the appropriate context.
  Therefore, in this section we establish our notation for these different R-matrices.

  First of all, the \textit{universal} R-matrix $R_{\text{univ}}$,
  which is a solution to the Quantum Yang-Baxter Equation (for details and notation see \cite{paper:MS-cs-theory-in-the-temporal-gauge})
  \begin{align}
    R_{\text{univ}} = \hat{P} q^{\sum_{\alpha \in \Delta} h_\alpha \otimes h_{\alpha^V}}
    \prod_{\beta \in \Phi^+}^\rightarrow \exp_{q_\beta} \left ( (q_\beta - q_\beta^{-1}) E_\beta \otimes F_\beta \right ),
  \end{align}
  which is valid universally, on the quantum group level.\footnote{That is,
  Chevalley generators $h_\alpha$, $E_\alpha$ and $F_\alpha$ are required to
  obey only the quantum group commutation relations plus Serre relations and nothing else.}

  Under a choice of particular irreducible representation $\lambda$ the universal
  $R$-matrix is mapped to specific $\text{dim}(\lambda)^2 \times \text{dim}(\lambda)^2$
  matrix.\footnote{It is also possible to specify two distinct representations
  $\lambda$ and $\nu$ for each of the two intertwined strands, but this is
  relevant for the case of links, which we do not consider in the present paper.}

  Then, one can consider an image of the universal R-matrix in particular
  highest weight representation $\lambda$ (with space $V_\lambda$),
  thus obtaining a matrix $R_{\lambda \otimes \lambda}$,
  which is a $\text{dim}(\lambda)^2 \times \text{dim}(\lambda)^2$ matrix.
  For instance, in case of the fundamental representation of $\mathcal{U}_q(sl(2))$
  this matrix is equal to
  \begin{align}
    R_{\Box \otimes \Box} =
    \left(\begin{array}{cccc}
      q & 0 & 0 & 0 \\
      0 & q-q^{-1} & 1 & 0 \\
      0 & 1 & 0 & 0 \\
      0 & 0 & 0 & -q^{-1}
    \end{array} \right)
  \end{align}

  A non-trivial theorem (see \cite{paper:R-invariants-of-tangles-I})
  is that operator $R_{\lambda \otimes \lambda}$ is diagonal in each
  irreducible representation $\nu \in \lambda\otimes\lambda$; so it makes
  sense to work with $R$-matrices for different $N$ all at once,
  we call such operator $\mathcal{R}_\lambda$
  (see \eqref{eq:r-matrix-klebsch-decomposition})

  Finally, consider a set of $R$-matrices relevant for $n$-strand tangle/braid calculus
  \begin{align} \label{eq:r-matrix-predef-picture}
    \begin{picture}(300,30)(0,-30)
      \thicklines
      \put(60,-13){$\R^{(i)}_\lambda = \ $}
      \put(100,0){
        \put(0,0){\line(0,-1){20}}
        \put(0,-22){$\scriptscriptstyle 1$}
        \put(20,0){\line(0,-1){20}}
        \put(20,-22){$\scriptscriptstyle 2$}
        \put(30,-10){$\dots$}
        \put(50,0){
          \put(20,0){\line(-1,-1){20}}
          \put(0,0){\qbezier(0,0)(8,-8)(8,-8)}
          \put(20,-20){\qbezier(0,0)(-8,8)(-8,8)}
        }
        \put(50,-22){$\scriptscriptstyle i$}
        \put(70,-22){$\scriptscriptstyle i+1$}
        \put(80,-10){$\dots$}
        \put(100,0){\line(0,-1){20}}
        \put(100,-22){$\scriptscriptstyle m$}
      }
      \put(230,-15){
        $= I_\lambda \otimes \cdots \otimes \mathcal{R}_{\lambda} \otimes \cdots \otimes I_\lambda$
        \put(-55,-7){$\scriptscriptstyle i, i+1$}
      }
    \end{picture}
  \end{align}
  where $I_\lambda$ is the identity operator in representation $\lambda$.

  Yet another non-trivial theorem is that if one considers irreducible representation
  decomposition
  \begin{align}
    V_\lambda^{\otimes n} = \sum_{\mu \vdash n|\lambda|} \mathcal{M}_{\lambda,\mu} V_\mu,
  \end{align}
  where $\mathcal{M}_{\lambda,\mu}$ are multiplicities of irreducible representation
  $\mu$ in the tensor power $\lambda^{\otimes n}$, then operators
  $\R^{(i)}_\lambda$ do not mix  different irreducible representation $\mu$;
  moreover, they are non-trivial only in the multiplicity (intertwiner) spaces
  and not in the representation spaces themselves -- where they are just diagonal operators,
  with known eigenvalues.

  Therefore, we call these $R$-matrices, acting in the multiplicity space of particular
  irreducible representation $\mu$ $\mathcal{R}_{\lambda,\mu}^{(i)}$.
  These are presicely the operators, entering the Reshetikhin-Turaev formula
  for colored HOMFLY polynomial \eqref{eq:homfly-definition}.

  See Section~\ref{sec:block-structure} for the more in-detail description of the block
  structure of the operators $\mathcal{R}^{(i)}_\lambda$ and $\mathcal{R}^{(i)}_{\lambda,\mu}$.
}

{\section{The multiplication of two one-hook Young diagrams} \label{sec:multiplication-of-two-1-hooks}
  In this and subsequent sections we always assume that we are working in the one-hook sector.
  For instance, all the statements about tensor products, even if not stated explicitly,
  should be understood to hold in the sense of the projection to the one-hook sector.

  One-hook (L-shape) Young diagram is the following collection of boxes.

  
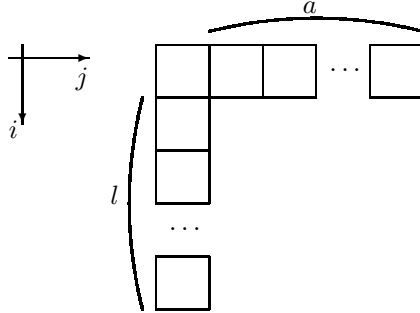
\begin{align} \label{eq:one-hook-partition}
    \begin{picture}(300,120)(0,-120)
      \put(50,0){
        \thinlines
        \put(0,0){\line(0,-1){30}}
        \put(0,-30){\vector(0,-1){0}}
        \put(-5,-5){\line(1,0){30}}
        \put(25,-5){\vector(1,0){0}}
        \put(-5,-35){$i$}
        \put(20,-15){$j$}
      }
      \put(100,0){
        \put(0,0){\boxPic}
        \put(0,-20){\boxPic}
        \put(0,-40){\boxPic}
        \put(0,-80){\boxPic}
        \put(20,0){\boxPic}
        \put(40,0){\boxPic}
        \put(80,0){\boxPic}
        \put(65,-10){$\dots$}
        \put(5,-70){$\dots$}
        \put(20,5){\thicklines \qbezier(0,0)(40,10)(80,0) \put(35,7){$a$}}
        \put(-5,-20){\thicklines \qbezier(0,0)(-10,-40)(0,-80) \put(-12,-40){$l$}}
      }
    \end{picture}
  \end{align}
  which we also denote by $(a, l)$. That is, 1-hook Young diagram $(a,l)$ has $a + l + 1$ boxes.

  The one-hook sector of the tensor product of two representations,
  corresponding to one-hook diagrams
  is decomposed into irreducibles in very simple way.
  In fact, we have the following

  \begin{lemma}\label{lem:multiplication-of-two-1-hooks}
    \begin{align}
      (a_1, l_1) \otimes (a_2, l_2) =\Bigg{|}_{\substack{\text{one hook}\\\text{sector}}}
      (a_1 + a_2 + 1, l_1 + l_2) \oplus (a_1 + a_2, l_1 + l_2 + 1)
    \end{align}
  \end{lemma}
  In this respect, one-hook representations are very similar to the fundamental representation,
  whose tensor square also has just two summands.
  \begin{proof}
    Direct application of Littlewood-Richardson rule
    \cite{paper:LR-coeffs}.
  \end{proof}

  The easy way to remember the multiplication rule for one-hook diagrams is to say
  that arm-boxes of both factors go into arms, leg-boxes go into legs and out of
  the two corner boxes one stays in the corner and the other can go either into the arm
  or into the leg of the result (which gives two summands).
}

{\section{The square of base representation and the structure of R-matrix}
  \label{sec:square-of-1-hook}

  A simple corollary of Lemma~\ref{lem:multiplication-of-two-1-hooks}
  is that for the square of our base representation $\lambda = (a, l)$ we have
  \begin{align}
    (a, l)^{\otimes 2} = (2 a + 1, 2 l) \oplus (2 a, 2 l + 1) =: A \oplus L,
  \end{align}
  that is we will denote the representation where the ``extra'' corner box went into an arm
  with $A$ and the representation where it went into a leg with $L$.

  Recall, that the eigenvalue of the $\R$-matrix, acting in some irreducible representation $Q \in \lambda^{\otimes 2}$ is
  (see Section~\ref{sec:r-matrix-defs} for our notations for different $R$-matrices)
  \begin{align} \label{eq:r-matrix-eigenvalue}
    e_Q = (-1)^\epsilon q^\varkappa,
  \end{align}
  where the exponent $\varkappa$ are manifestly given by
  (see, for instance, \cite{paper:R-invariants-of-tangles-I}, p.23, eq.~(2.14))
  \begin{align}
    \varkappa & = \sum_{(i,j) \in Q} (j - i),
  \end{align}
  where the sum runs over boxes of the Young diagram $Q$ (see \eqref{eq:one-hook-partition}
  for the directions of axes $i$ and $j$).
  The sign $\epsilon$ is determined by the parity of the representation $Q$ inside $\lambda^{\otimes 2}$. Here we are working in the so-called \textit{vertical framing},
  which more prominently highlights relation of knot invariants with representation theory.
  Other choices of framing, in particular, \textit{topological framing} are possible,
  that highlight other aspects (e.g. topological invariance of the obtained quantities)
  \cite{paper:AM-cabling-procedure}.

  Applying this to the situation at hand (one-hook sector of the the square of the base representation)
  we get
  \begin{lemma}\label{lem:tensor-square-eigenvalues}
    \begin{align} \label{eq:tensor-square-eigenvalues}
    e_A & = (-1)^l q^{(2 a^2 + 2 a - 2 l^2 - 2 l) + (a + l + 1)}
    \\ \notag
    e_L & = (-1)^{l+1} q^{(2 a^2 + 2 a - 2 l^2 - 2 l) - (a + l + 1)},
    \end{align}
    where we've grouped terms in the exponent of $q$ to make the formula more similar
    to the fundamental representation.
  \end{lemma}
  \begin{proof}
    The exponent of $q$ is obtained by direct evaluation of simple arithmetic progressions.

    The signs require more care. The parity of representation $Q$ inside $\lambda \otimes \lambda$
    depends on whether Schur polynomial $s_Q(p)$ occurs in the Schur polynomial decomposition
    of $s_\lambda(p)^2 + s_\lambda(p^{(2)})$ or $s_\lambda(p)^2 - s_\lambda(p^{(2)})$
    (see, \cite{paper:AM-cabling-procedure}). Here, $p^{(2)}$ denotes the Adams operation.
    Since Schur polynomials are characters, and the decomposition mirrors the decomposition
    of a linear space into irreducible summands, we know that the decomposition coefficients,
    if non-zero, should be strictly positive and integral.

    It is straightforward to see that terms, linear in times $p_k$,
    occur only in Schur polynomials corresponding to one-hook diagrams.
    Indeed, Schur polynomial $s_\lambda$ is a (generalized) homogeneous polynomial
    of degree $|\lambda|$, where $p_k$ has weight $k$.
    Therefore, Schur polynomial's linear term is $p_{|\lambda|}$, with some (maybe, vanishing) coefficient.

    At the same time, the Jacobi-Trudi formula expresses
    any Schur polynomial as determinant of a matrix, made of Schur polynomials
    for symmetric representations (see, for instance, \cite{book:M-symmetric-functions-and-hall-polynomials})
    \begin{align} \label{eq:schur-jacobi-trudi}
      s_\lambda (p) = \det_{l(\lambda)\times l(\lambda)} s_{[\lambda_i - i + j]} (p), \ \ \text{where}
      \sum_{m=0}^\infty z^m s_{[\lambda_i - i + j]} (p) := \exp\left(\sum_k \frac{p_k z^k}{k}\right),
    \end{align}
    and only matrices for the one-hook Schur polynomials contain $s_{[|\lambda|]}$ entry.
    
    Furthermore, from \eqref{eq:schur-jacobi-trudi} one gets, that linear term
    of Schur polynomial for one-hook diagram $(a,l)$ is equal to
    \begin{align}
      (-1)^{l} \frac{1}{a + l + 1} p_{a + l + 1},
    \end{align}

    So, if $l$ is even, then $s_{(2 a + 1, 2 l)}(p)$ occurs in the decomposition of $s_{(a,l)}(p)^2 + s_{(a,l)}(p^{(2)})$
    and $s_{(2 a, 2 l + 1)}(p)$ -- in the decomposition of $s_{(a,l)}(p)^2 - s_{(a,l)}(p^{(2)})$.
    If $l$ is odd, the situation is reversed, and we get the statement of the lemma.
  \end{proof}
}

{\section{The block structure of the R-matrices} \label{sec:block-structure}
  While the problem of quick and convenient calculation of arbitrary matrices $\R_{\lambda,\mu}^{(i)}$
  from Young graph is still open, it is relatively easy to establish the block structure
  of these $\R$-matrices from the Young graph.

  What we review in this section is well-known to the experts and is, in fact,
  a simple combination of lemmas from \cite{paper:R-invariants-of-tangles-I}.

  First of all, one of the possible basis choices in the multiplicity space $(\lambda,\mu)$, that is,
  the multiplicity of irreducible representation $\mu$ in the $m$th tensor power of
  representation $\lambda$, where $m$ is the number of strands,
  is given by comb diagrams (\cite{paper:R-invariants-of-tangles-I}, Proposition~3.3, p.28)
  that on the language of the Young graph correspond to paths
  
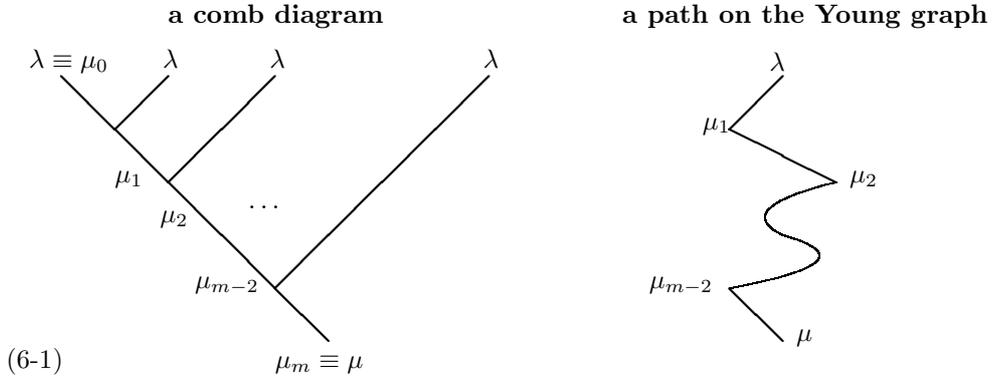
\begin{align} \label{eq:comb-basis}
    \begin{picture}(300,150)(0,-150)
      \thicklines
      \put(-20,-40){
        \put(40,20){\textbf{a comb diagram}}
        \put(-12,3){$\lambda \equiv \mu_0$}
        \put(0,0){\line(1,-1){100}}
        \put(38,3){$\lambda$}
        \put(20,-20){\line(1,1){20}}
        \put(20,-40){$\mu_1$}
        \put(78,3){$\lambda$}
        \put(40,-40){\line(1,1){40}}
        \put(158,3){$\lambda$}
        \put(37,-55){$\mu_2$}
        \put(80,-80){\line(1,1){80}}
        \put(70,-50){$\dots$}
        \put(50,-80){$\mu_{m-2}$}
        \put(80,-110){$\mu_{m} \equiv \mu$}
      }
      \put(250,-40){
        \put(-60,20){\textbf{a path on the Young graph}}
        \put(0,0){\line(-1,-1){20}}
        \put(-5,2){$\lambda$}
        \put(-20,-20){\line(2,-1){40}}
        \put(-30,-20){$\mu_1$}
        \put(-20,-80){\line(1,-1){20}}
        \put(25,-40){$\mu_2$}
            {
             \linethickness{0.1mm} 
          \put(20,-40){\qbezier(0,0)(-40,-10)(-20,-20)}
          \put(-50,-80){$\mu_{m-2}$}
          \put(20,-40){\qbezier(-20,-20)(15,-30)(-40,-40)}
        }
        \put(5,-100){$\mu$}
      }
    \end{picture}
  \end{align}
  The comb-shape of the diagram is fixed and different basis elements are distinguished
  by the partitions $\mu_1, \dots \mu_m$, assigned to the edges of the comb diagram.

  The dual basis is obtained by reflecting comb diagrams with respect to the horizontal line.
  Indeed, from the fact that the simplest comb diagrams, corresponding to Clebsch-Gordan coefficients,
  are orthogonal (\cite{paper:R-invariants-of-tangles-I}, formula~(2.27), p.25),
  \begin{align} \label{eq:klebsch-orthogonality}
    \begin{picture}(300,80)(0,-80)
      \thicklines
      \put(80,0){
        \put(0,0){\line(0,-1){15}}
        \put(0,-30){\circle{30}}
        \put(0,-45){\line(0,-1){15}}
        \put(-10,-10){$\nu_1$}
        \put(-10,-55){$\nu_2$}
        \put(-30,-30){$\mu_1$}
        \put(20,-30){$\mu_2$}
      }
      \put(120,-30){$=$}
      \put(150,0){
        \put(0,0){\line(0,-1){60}}
        \put(-12,-30){$\nu_1$}
        \put(5,-30){$\delta_{\nu_1,\nu_2}$}
      }
    \end{picture}
  \end{align}
  by induction, one readily sees that the reflected comb diagrams do form a dual basis.
  
  The matrix $\R_\lambda$ acts diagonally in each irreducible component
  of the square of base representation $\lambda$
  (\cite{paper:R-invariants-of-tangles-I}, formulas~(2.14) and~(2.15), p.23)
  \begin{align} \label{eq:r-matrix-klebsch-decomposition}
    \begin{picture}(300,60)(0,-60)
      \thicklines
      \put(100,-20){
        \put(0,0){
          \put(20,0){\line(-1,-1){20}}
          \put(0,0){\qbezier(0,0)(8,-8)(8,-8)}
          \put(20,-20){\qbezier(0,0)(-8,8)(-8,8)}
          \put(-5,2){$\lambda$}
          \put(20,0){\put(-5,2){$\lambda$}}
          \put(0,-20){\put(-5,-8){$\lambda$}}
          \put(20,-20){\put(-5,-8){$\lambda$}}
        }
        \put(30,-13){$=\ \sum\limits_{\nu \in \lambda^{\otimes 2}} e_\nu$}
        \put(100,0){
          \put(-3,2){$\lambda$}
          \put(0,0){\line(2,-1){10}}
          \put(17,2){$\lambda$}
          \put(20,0){\line(-2,-1){10}}
          \put(-3,-30){$\lambda$}
          \put(10,-5){\line(0,-1){10}}
          \put(17,-30){$\lambda$}
          \put(0,-20){\line(2,1){10}}
          \put(20,-20){\line(-2,1){10}}
          \put(12,-12){$\nu$}
        }
      }
    \end{picture}
  \end{align}
  where the eigenvalue $e_\nu$ is given by formula~\eqref{eq:r-matrix-eigenvalue}.

  Other basis choices for multiplicity space are possible. In particular,
  one can consider another, different from a comb,
  trivalent tree with root $\mu$ and $m$ leaves $\lambda$. Different basis elements
  are again distinguished by labeling intermediate edges with partitions.

  For the multiplicity space of the tensor product of three representations
  there are two possible basis choices -- the comb one and the other --
  that are related by the Racah-Wigner matrix
  \begin{align} \label{eq:racah-definition}
    \begin{picture}(300,100)(60,-100)
      \thicklines
      \put(30,-10){
        \put(0,0){\line(1,-1){60}}
        \put(20,-20){\line(1,1){20}}
        \put(40,-40){\line(1,1){40}}
        \put(-3,5){$\rho$}
        \put(37,5){$\sigma$}
        \put(77,5){$\eta$}
        \put(57,-70){$\xi$}
        \put(22,-40){$\zeta$}
      }
      \put(130,-40){
        $=\ \ \sum_{\omega \in \sigma \otimes \eta} U\left(\begin{array}{ccc} \rho & \sigma & \eta \\ \ & \xi &\  \end{array} \right)^\zeta_{\ \omega}$
      }
      \put(280,-10){
        \put(0,0){\line(1,-1){60}}
        \put(40,0){\line(1,-1){20}}
        \put(40,-40){\line(1,1){40}}
        \put(-3,5){$\rho$}
        \put(37,5){$\sigma$}
        \put(77,5){$\eta$}
        \put(57,-70){$\xi$}
        \put(51,-37){$\omega$}
      }
    \end{picture}
  \end{align}
  The Racah-Wigner matrix depends on three input representations and one output representation, and has representations
  in two possible choices of intermediate channel as indices.
  
  Each matrix $\R^{(i)}$, which in this pictorial
  language is represented as
  (see \cite[Proposition~3.1,~p.27]{paper:R-invariants-of-tangles-I})
  \begin{align} \label{eq:r-matrix-picture}
    \begin{picture}(300,30)(0,-30)
      \thicklines
      \put(60,-13){$\R^{(i)} = \ $}
      \put(100,0){
        \put(0,0){\line(0,-1){20}}
        \put(0,-22){$\scriptscriptstyle 1$}
        \put(20,0){\line(0,-1){20}}
        \put(20,-22){$\scriptscriptstyle 2$}
        \put(30,-10){$\dots$}
        \put(50,0){
          \put(20,0){\line(-1,-1){20}}
          \put(0,0){\qbezier(0,0)(8,-8)(8,-8)}
          \put(20,-20){\qbezier(0,0)(-8,8)(-8,8)}
        }
        \put(50,-22){$\scriptscriptstyle i$}
        \put(70,-22){$\scriptscriptstyle i+1$}
        \put(80,-10){$\dots$}
        \put(100,0){\line(0,-1){20}}
        \put(100,-22){$\scriptscriptstyle m$}
      }
    \end{picture}
  \end{align}
  where with small scripts we denoted the strands' numbers.

  Therefore, combining \eqref{eq:klebsch-orthogonality}, \eqref{eq:r-matrix-klebsch-decomposition}
  and \eqref{eq:r-matrix-picture} we see that  $\R^{(i)}$ is diagonal in the basis
  \begin{align} \label{eq:i-th-basis}
    \begin{picture}(300,100)(-50,-100)
      \thicklines
      \put(0,0){\line(1,-1){100}}
      \put(0,0){$\scriptscriptstyle 1$}
      \put(20,0){\line(-1,-1){10}}
      \put(30,0){\line(-1,-1){15}}
      \put(20,0){$\scriptscriptstyle 2$}
      \put(30,0){$\scriptscriptstyle 3$}
      \put(20,-15){$\dots$}
      \put(50,0){\line(-1,-1){25}}
      \put(40,-40){\line(1,1){30}}
      \put(45,0){$\scriptscriptstyle i-1$}
      \put(60,0){$\scriptscriptstyle i$}
      \put(72,0){$\scriptscriptstyle i+1$}
      \put(87,0){$\scriptscriptstyle i+2$}
      \put(70,-10){\line(-1,1){10}}
      \put(70,-10){\line(1,1){10}}
      \put(90,0){\line(-1,-1){45}}
      \put(60,-55){$\dots$}
      \put(80,-80){\line(1,1){80}}
      \put(150,0){\line(-1,-1){75}}
      \put(142,0){$\scriptscriptstyle m-1$}
      \put(160,0){$\scriptscriptstyle m$}
    \end{picture}
  \end{align}

  Moreover, because both $i$th basis \eqref{eq:i-th-basis} and the comb-basis \eqref{eq:comb-basis}
  are orthogonal from \eqref{eq:racah-definition} we see that
  the transition matrix between bases \eqref{eq:i-th-basis} and \eqref{eq:comb-basis}
  depends only on partitions in the intermediate channels on layers $i-1$ and $i+1$.

  Hence, each $\R^{(i)}$ consists of blocks.
  Each block acts in the subspace spanned by paths on the Young graph that coincide
  up to and including level $i-1$ and from level $i+1$ to $m$, and differ only by one
  partition on level $i$:
  \begin{align} \label{eq:paths-block}
    \begin{picture}(300,100)(-130,-100)
      \thicklines
      \put(0,0){\qbezier(0,0)(-10,-10)(0,-20)}
      \put(2,-2){$\lambda$}
      \put(0,-20){\qbezier(0,0)(10,-10)(0,-20)}
      \put(0,-40){\put(-15,2){$\scriptstyle \mu_{i-1}$}}
      \put(0,-40){\line(-2,-1){20}}
      \put(0,-40){\line(-1,-1){10}}
      \put(0,-40){\line(1,-1){10}}
      \put(0,-40){\line(2,-1){20}}
      \put(0,-60){\put(2,-5){$\scriptstyle \mu_{i+1}$}}
      \put(0,-60){\line(-2,1){20}}
      \put(0,-60){\line(-1,1){10}}
      \put(0,-60){\line(1,1){10}}
      \put(0,-60){\line(2,1){20}}
      \put(0,-50){\put(22,-2){$\mu_i;$ \ \ $\mu_i$ is arbitrary}}
      \put(0,-60){\qbezier(0,0)(-10,-10)(0,-20)}
      \put(0,-80){\qbezier(0,0)(10,-10)(0,-20)}
      \put(0,-100){\put(-8,-2){$\mu$}}
    \end{picture}
  \end{align}
  The size of the block is therefore equal to the number of
  partitions (with multiplicities) on level $i$, through which a path can pass, given its begining and end.
  For instance, the drawn bundle of paths \eqref{eq:paths-block}
  forms a $4 \times 4$ block, since
  there are $4$ possible choices for partition $\mu_i$ on level $i$.
  
  In Section~\ref{sec:one-hook-young-graph} we show that
  for the Young graph we are interested in  each matrix $\R^{(i)}$
  consists only from $1 \times 1$ (singlet) and $2 \times 2$ (doublet) blocks.
}

{\section{The one-hook Young graph}\label{sec:one-hook-young-graph}

  The Young graph that we are interested in is very special.
  First of all, our base representation $\lambda$ (which provides transition between the layers of the graph)
  is one-hook.
  Second, because we are considering the Alexander polynomial, all the summands
  in the Reshetikhin-Turaev formula \eqref{eq:homfly-definition}
  that correspond to non-one-hook (fat-hook) Young diagrams vanish
  (see Proposition~\ref{prop:one-hook-sufficient}
  and also Section~\ref{sec:topological-schurs} for more details).
  Therefore, we are interested in the projection of the Young graph
  onto the one-hook sector.

  Therefore, we have the following
  \begin{lemma}\label{lem:young-graph-structure}
    The Young graph that is sufficient for calculating the colored Alexander polynomials
    for one-hook Young diagrams has the form of a quadrant of a plane
    
    
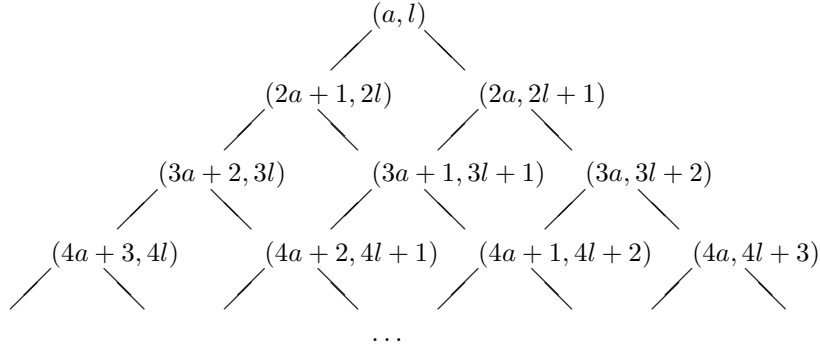
\begin{align}
      \begin{picture}(300,140)(0,-140)
        \put(120,0){
          \put(0,0){
            \put(0,0){$(a,l)$}
            \put(0,-3){\line(-1,-1){15}}
            \put(20,-3){\line(1,-1){15}}
          }
          \put(-40,-30){
            \put(0,0){$(2 a + 1,2 l)$}
            \put(0,-3){\line(-1,-1){15}}
            \put(20,-3){\line(1,-1){15}}
          }
          \put(40,-30){
            \put(0,0){$(2 a,2 l + 1)$}
            \put(0,-3){\line(-1,-1){15}}
            \put(20,-3){\line(1,-1){15}}
          }
          \put(-80,-60){
            \put(0,0){$(3 a + 2,3 l)$}
            \put(0,-3){\line(-1,-1){15}}
            \put(20,-3){\line(1,-1){15}}
          }
          \put(0,-60){
            \put(0,0){$(3 a + 1,3 l + 1)$}
            \put(0,-3){\line(-1,-1){15}}
            \put(20,-3){\line(1,-1){15}}
          }
          \put(80,-60){
            \put(0,0){$(3 a,3 l + 2)$}
            \put(0,-3){\line(-1,-1){15}}
            \put(20,-3){\line(1,-1){15}}
          }
          \put(-120,-90){
            \put(0,0){$(4 a + 3,4 l)$}
            \put(0,-3){\line(-1,-1){15}}
            \put(20,-3){\line(1,-1){15}}
          }
          \put(-40,-90){
            \put(0,0){$(4 a + 2,4 l + 1)$}
            \put(0,-3){\line(-1,-1){15}}
            \put(20,-3){\line(1,-1){15}}
          }
          \put(40,-90){
            \put(0,0){$(4 a + 1,4 l + 2)$}
            \put(0,-3){\line(-1,-1){15}}
            \put(20,-3){\line(1,-1){15}}
          }
          \put(120,-90){
            \put(0,0){$(4 a,4 l + 3)$}
            \put(0,-3){\line(-1,-1){15}}
            \put(20,-3){\line(1,-1){15}}
          }
          \put(0,-120){$\dots$}
        }
      \end{picture}
    \end{align}
  \end{lemma}
  \begin{proof}
    By induction. The first layer of the graph consists only from partition $(a,l)$,
    so in particular, it has the desired structure.

    Now, suppose that the graph has the desired structure up to and including layer $n$.
    Consider some partition $(n a + \delta_A, n l + \delta_L)$ on layer $n$.
    From Lemma~\ref{lem:multiplication-of-two-1-hooks} it connects to exactly
    two partitions on the layer $n+1$, namely, to $((n+1) a + \delta_A + 1, (n+1) l + \delta_L)$
    and to $((n+1) a + \delta_A, (n+1) l + \delta_L + 1)$.

    Thus, layer $n+1$ also has the desired structure and connections from layer $n$ to layer $n+1$
    are as they should be.
  \end{proof}
}

{\section{Singlet and doublet blocks of the R-matrix}
  
  From the explicit structure of the Young graph (see Section \ref{sec:one-hook-young-graph})
  it is straightforward to see, that all $\R$-matrices
  consist only from $1 \times 1$ and $2 \times 2$ blocks
  (see Section \ref{sec:background} and also
  \cite{paper:AM-cabling-procedure,
    paper:AMMM-colored-homfly-polynomials-as-multiple-sums}).
  The situation in this respect is similar
  to the case of fundamental representation.

  Indeed, if one considers two representation $P$ and $Q$ on layers $n-1$ and $n+1$ of the graph, respectively,
  then one of the three situations can happen:

  \begin{enumerate}
  \item There are no paths between them. In this case paths passing through $P$ and $Q$ never belong to the same block.
  \item There is exactly one path between them.
    These are the cases
    
    \begin{align}
      \begin{picture}(300,100)(0,-100)
        \put(0,-15){
          \put(0,15){The ``arm'' case:}
          \put(-20,0){
            \put(0,0){$((n-1) a + \delta_a, (n-1) l + \delta_l)$}
            \put(60,-5){\line(0,-1){15}}
          }
          \put(0,-30){
            \put(0,0){$(n a + \delta_a \mathbf{+ 1}, n l + \delta_l)$}
            \put(40,-5){\line(0,-1){15}}
          }
          \put(-30,-60){
            $((n+1) a + \delta_a \mathbf{+ 2}, (n+1) l + \delta_l)$
          }
        }
        \put(170,-15){
          \put(0,15){The ``leg'' case:}
          \put(-20,0){
            \put(0,0){$((n-1) a + \delta_a, (n-1) l + \delta_l)$}
            \put(60,-5){\line(0,-1){15}}
          }
          \put(0,-30){
            \put(0,0){$(n a + \delta_a , n l + \delta_l \mathbf{+ 1})$}
            \put(40,-5){\line(0,-1){15}}
          }
          \put(-30,-60){
            $((n+1) a + \delta_a, (n+1) l + \delta_l \mathbf{+ 2})$
          }
        }
      \end{picture}
    \end{align}
    
    That is, in the former case the extra box is added two times to the arm,
    while in the latter case -- two times to the leg.
    In both cases the resulting block of the R-matrix $\R^{(n)}$ is $1 \times 1$, or in other words \textit{singlet},
    and the corresponding eigenvalue is $e_A$ or $e_L$ (see equation~\eqref{eq:tensor-square-eigenvalues}), respectively.

  \item There are two paths between $P$ and $Q$.

    \begin{align}
      \begin{picture}(300,80)(0,-80)
        \put(100,0){
          \put(0,0){
            \put(0,0){$((n-1) a + \delta_a, (n-1) l + \delta_l)$}
            \put(55,-5){\line(-1,-1){15}}
            \put(75,-5){\line(1,-1){15}}
          }
          \put(-40,-30){
            \put(0,0){$(n a + \delta_a \mathbf{+ 1}, n l + \delta_l)$}
          }
          \put(80,-30){
            \put(0,0){$(n a + \delta_a, n l + \delta_l \mathbf{+ 1})$}
          }
          \put(-10,-60){
            \put(0,0){$((n+1) a + \delta_a \mathbf{+ 1}, (n+1) l + \delta_l \mathbf{+ 1})$}
            \put(65,10){\line(-1,1){15}}
            \put(85,10){\line(1,1){15}}
          }
        }
      \end{picture}
    \end{align}
    which give a $2 \times 2$ block in the R-matrix $\R^{(n)}$ for every pair of paths
    that coincide everywhere, except the piece between layers $(n-1)$ and $(n+1)$.
  \end{enumerate}

  From formula \eqref{eq:tensor-square-eigenvalues} we see, that (up to a common framing factor)
  the eigenvalues $e_A$ and $e_L$ are very similar to the eigenvalues $q$ and $-1/q$ for the fundamental representation,
  except the role of $q$ is played by $q^{|R|}$. The goal of subsequent sections is to rigorously show
  that this is, indeed, the case: while expressions for R-matrices $\R^{(n)}$ are complicated, they
  are nevertheless obtained from the expressions for the fundamental case by simple substitution
  $q \rightarrow q^{|R|},$ up to a common framing factor.
}

{\section{The first 2 by 2 block} \label{sec:first-two-by-two-block}

  In this section we establish the formula for the very first doublet that occurs
  in the Young graph. Namely, the one, corresponding to the square

  \begin{align} \label{eq:picture-of-first-square}
    \begin{picture}(300,60)(0,-60)
      \put(100,0){
        \put(0,0){
          \put(25,0){$(a, l)$}
          \put(25,-5){\line(-1,-1){15}}
          \put(45,-5){\line(1,-1){15}}
        }
        \put(-60,-30){
          \put(40,0){$(2 a \mathbf{+ 1}, 2 l)$}
        }
        \put(40,-30){
          \put(5,0){$(2 a, 2 l \mathbf{+ 1})$}
        }
        \put(-20,-60){
          \put(25,0){$(3 a \mathbf{+ 1}, 3 l \mathbf{+ 1})$}
          \put(45,10){\line(-1,1){15}}
          \put(65,10){\line(1,1){15}}
        }
      }
    \end{picture}
  \end{align}
  
  This will provide the base for the recursion
  (described in section~\ref{sec:two-recursions-for-doublets})
  to find explicit form of all other doublets ($2 \times 2$ blocks).

  Namely we have the following
  \begin{lemma}\label{lem:the-very-first-doublet}
    The R-matrix block, corresponding to the piece of Young graph
    $(a,l) \rightarrow [(2 a + 1, 2 l), (2 a, 2 l + 1)] \rightarrow (3 a + 1, 3 l + 1)$
    is equal to
    \begin{align}
      B_{\substack{\text{first} \\ \text{doublet}}} & =
      \frac{q^{2 a^2 + 2 a - 2 l^2 - 2l}}{[2]_\bullet}
      \left(\begin{array}{cc}
        (-1)^{l+1} q^{-2 a - 2 - 2 l} & \sqrt{[3]_\bullet} \\
        \sqrt{[3]_\bullet} & (-1)^l q^{2 a + 2 + 2 l},
      \end{array}\right)
    \end{align}
    where the $q$-number with a bullet is the usual $q$-number $[n] = (q^n - q^{-n})/(q - q^{-1})$
    with $q$ substituted by $q^{a + l + 1}$:
    \begin{align}
      [n]_\bullet = \frac{q^{n (a + l + 1)} - q^{-n (a + l + 1)}}{q^{a + l + 1} - q^{- (a + l + 1)}}
    \end{align}
  \end{lemma}
  \begin{proof}
    This lemma is instructive, as the logic of the proof of the recursion steps would be very similar.

    Let the first double block be some unknown $2 \times 2$ symmetric matrix (since R-matrix is always obtained
    from symmetric matrix with help of the unitary matrix, see \cite{paper:AM-cabling-procedure}).
    
    \begin{align}
      B_{\substack{\text{first} \\ \text{doublet}}} & =
      \left(\begin{array}{cc}
        r_{11} & r_{12} \\ r_{12} & r_{22}
      \end{array}\right)
    \end{align}

    We can immediately write two equations for the matrix elements
    \begin{align}\label{eq:equation-from-traces}
      \tr B_{\substack{\text{first} \\ \text{doublet}}} = & \ e_A + e_L = r_{11} + r_{22} \\ \notag
      \tr B^2_{\substack{\text{first} \\ \text{doublet}}} = & \ e_A^2 + e_L^2 = r_{11}^2 + r_{22}^2 + 2 r_{12}^2,
    \end{align}
    since we know that the $2 \times 2$ block has eigenvalues $e_A$ and $e_L$.

    Therefore, what we need is the third equation. We use the Yang-Baxter equation between
    matrices $\R^{(1)}$ and $\R^{(2)}$ in the form that does not contain inverse matrices and
    hence is most convenient for our purposes
    \begin{align} \label{eq:our-yang-baxter}
      \begin{picture}(300,120)(-40,-120)
        \thicklines
        \put(150,0){
          \put(0,0){
            \put(0,0){\line(0,-1){20}}
            \put(20,0){\line(0,-1){20}}
            \put(40,0){\line(0,-1){20}}
          }
          \put(0,-20){
            \put(0,0){\qbezier(0,0)(1,-1)(7,-7)}
            \put(20,-20){\qbezier(0,0)(-1,1)(-7,7)}
            \put(20,0){\line(-1,-1){20}}
            \put(40,0){\line(0,-1){20}}
          }
          \put(0,-40){
            \put(0,0){\line(0,-1){20}}
            \put(20,0){\qbezier(0,0)(1,-1)(7,-7)}
            \put(40,-20){\qbezier(0,0)(-1,1)(-7,7)}
            \put(40,0){\line(-1,-1){20}}
          }
          \put(0,-60){
            \put(0,0){\qbezier(0,0)(1,-1)(7,-7)}
            \put(20,-20){\qbezier(0,0)(-1,1)(-7,7)}
            \put(20,0){\line(-1,-1){20}}
            \put(40,0){\line(0,-1){20}}
          }
          \put(0,-80){
            \put(0,0){\line(0,-1){20}}
            \put(20,0){\line(0,-1){20}}
            \put(40,0){\line(0,-1){20}}
          }
          \put(0,0){
            \put(50,0){\line(0,-1){100}}
            \put(60,-50){$\dots$}
            \put(80,0){\line(0,-1){100}}
          }
          \put(-15,-120){
            $\R^{(2)} \ \R^{(1)} \ \R^{(2)}$
          }
        }
        \put(0,0){
          \put(0,0){
            \put(0,0){\line(0,-1){20}}
            \put(20,0){\line(0,-1){20}}
            \put(40,0){\line(0,-1){20}}
          }
          \put(0,-20){
            \put(0,0){\line(0,-1){20}}
            \put(20,0){\qbezier(0,0)(1,-1)(7,-7)}
            \put(40,-20){\qbezier(0,0)(-1,1)(-7,7)}
            \put(40,0){\line(-1,-1){20}}
          }
          \put(0,-40){
            \put(0,0){\qbezier(0,0)(1,-1)(7,-7)}
            \put(20,-20){\qbezier(0,0)(-1,1)(-7,7)}
            \put(20,0){\line(-1,-1){20}}
            \put(40,0){\line(0,-1){20}}
          }
          \put(0,-60){
            \put(0,0){\line(0,-1){20}}
            \put(20,0){\qbezier(0,0)(1,-1)(7,-7)}
            \put(40,-20){\qbezier(0,0)(-1,1)(-7,7)}
            \put(40,0){\line(-1,-1){20}}
          }
          \put(0,-80){
            \put(0,0){\line(0,-1){20}}
            \put(20,0){\line(0,-1){20}}
            \put(40,0){\line(0,-1){20}}
          }
          \put(0,0){
            \put(50,0){\line(0,-1){100}}
            \put(60,-50){$\dots$}
            \put(80,0){\line(0,-1){100}}
          }
          \put(-15,-120){
            $\R^{(1)} \ \R^{(2)} \ \R^{(1)}$
          }
        }
        \put(95,-120){$=$}
        \put(95,-50){$=$}
      \end{picture}
    \end{align}
    
    Since $\R^{(1)}$ is diagonal and manifestly given by (for the paths belonging to this $2 \times 2$ block)
    \begin{align}
      \R^{(1)} =
      \left( \begin{array}{cc}
        e_A & 0 \\
        0 & e_L
      \end{array} \right )
    \end{align}
    the Yang-Baxter equation~\eqref{eq:our-yang-baxter}, written in components, reads
    \begin{align} \label{eq:yang-baxter-in-components}
      \left( \begin{array}{cc}
        e_A r_{11}^2 + e_L r_{12}^2
        & e_A r_{11} r_{12} + e_L r_{12} r_{22} \\
        e_A r_{11} r_{12} + e_L r_{12} r_{22} & e_A r_{12}^2 + e_L r_{22}^2
      \end{array} \right )
      =
      \left( \begin{array}{cc}
        e_A^2 r_{11} & e_A e_L r_{12} \\
        e_A e_L r_{12} & e_L^2 r_{22}
      \end{array} \right ),
    \end{align}
    Without loss of generality, consider the 1,1 matrix element of this matrix equation.
    This provides the desired third equation.
    
    Solving equations~\eqref{eq:equation-from-traces} and~\eqref{eq:yang-baxter-in-components}
    for $r_{11}$, $r_{12}$ and $r_{22}$ we obtain the statement of the lemma.
  \end{proof}
}

{\section{Two recursions on the Young graph for the 2 by 2 blocks} \label{sec:two-recursions-for-doublets}
  In this section we are interested in an explicit formula for arbitrary $2 \times 2$ block.

  First, let us informally sketch the logic. In the previous Section (Lemma~\ref{lem:the-very-first-doublet})
  we've obtained the formula for the first $2 \times 2$ block.
  Since our Young graph is a rectangular grid (Lemma~\ref{lem:young-graph-structure}),
  its arbitrary $2 \times 2$ block can be reached from the first one
  \begin{align}
    \begin{picture}(300,60)(0,-60)
      \put(130,0){
        \multiput(0,0)(-10,-10){3}{\multiput(0,0)(10,-10){4}{\diagBoxPic}}
        \put(0,0){
          \thicklines
          \put(0,0){\vector(-1,-1){20}}
          \put(-20,-20){\vector(1,-1){30}}
        }
      }
    \end{picture}
  \end{align}
  provided we know two recursion steps
  \begin{align}
    \begin{picture}(300,40)(0,-30)
      \put(50,0){
        \put(-50,5){The ``arm'' recursion step}
        \put(0,0){\line(1,-1){10}}
        \put(-10,-10){\line(1,-1){10}}
        \put(-20,-20){\line(1,-1){10}}
        \put(10,-10){\line(-1,-1){20}}
        \put(0,0){\line(-1,-1){20}}
        \put(0,-10){\thicklines \vector(-1,-1){10}}
      }
      \put(220,0){
        \put(-50,5){The ``leg'' recursion step}
        \put(0,0){\line(-1,-1){10}}
        \put(10,-10){\line(-1,-1){10}}
        \put(20,-20){\line(-1,-1){10}}
        \put(-10,-10){\line(1,-1){20}}
        \put(0,0){\line(1,-1){20}}
        \put(0,-10){\thicklines \vector(1,-1){10}}
      }
    \end{picture}
  \end{align}

  Of course, there are multiple recursion paths (Manhattan walks) leading from the first square
  to a given one. They all give the same answer, since correctness of the definition
  of colored HOMFLY-PT polynomial is established in \cite{paper:R-invariants-of-tangles-I}.
  
  Therefore, we have the following
  \begin{lemma} \label{lem:arbitrary-2-by-2-block}
    The expression of the $2 \times 2$ block depends only on the level of the Young graph
    which this block occurs on and does not depend on the horizontal position of the block in the level.
    Moreover, the block is explicitly given by
    \begin{align}
      B^{(n)}_{2 \times 2} = \frac{q^{2 a^2 + 2 a - 2 l^2 - 2l}}{[n]_\bullet}
      \left ( \begin{array}{cc}
        (-1)^{l+1} q^{-2 a - 2 l - 2} & \sqrt{[n+1]_\bullet} \\
        \sqrt{[n+1]_\bullet} & (-1)^{l} q^{2 a + 2 l + 2}
      \end{array}
      \right )
    \end{align}
    that is $B_{\substack{\text{first} \\ \text{doublet}}}$
    from  Lemma~\ref{lem:the-very-first-doublet}
    is $B^{(2)}_{2 \times 2}$.
  \end{lemma}
  \begin{proof}
    First of all, the eigenvalues of arbitrary $2 \times 2$ block are $e_A$ and $e_L$,
    so equations~\eqref{eq:equation-from-traces} do hold for arbitrary block.
    We need the third equation, and, as in the case of the base of the recursion,
    we will use Yang-Baxter equation to get it.

    We are interested in finding an expression for the $2 \times 2$ block on level $n$,
    assuming we know expressions for all the $2 \times 2$ blocks up to level $n-1$.
    Consider the following piece of the Young graph
    \begin{align} \label{eq:leg-step-subgraph}
      \begin{picture}(300,140)(0,-110)
        \thicklines
        \put(100,-20){
          \put(-70,30){
            \put(0,0){$((n-2) a + \delta_a, (n-2) l + (\delta_l - 1))$}
            \put(45,-5){\line(-1,-1){15}}
            \put(115,-5){\line(1,-1){15}}
          }
          \put(-180,0){
            \put(0,0){$((n-1) a + (\delta_a + 1), (n-1) l + (\delta_l - 1))$}
            \put(145,-5){\line(1,-1){15}}
          }
          \put(0,0){
            \put(0,0){$((n-1) a + \delta_a, (n-1) l + \delta_l)$}
            \put(25,-5){\line(-1,-1){15}}
            \put(75,-5){\line(1,-1){15}}
          }
          \put(-60,-30){
            \put(0,0){$(n a + \delta_a \mathbf{+ 1}, n l + \delta_l)$}
          }
          \put(80,-30){
            \put(0,0){$(n a + \delta_a, n l + \delta_l \mathbf{+ 1})$}
          }
          \put(-20,-60){
            \put(0,0){$((n+1) a + \delta_a \mathbf{+ 1}, (n+1) l + \delta_l \mathbf{+ 1})$}
            \put(45,10){\line(-1,1){15}}
            \put(95,10){\line(1,1){15}}
          }
        }
      \end{picture}
    \end{align}

    There are three paths, joining the top and the bottom vertex of this piece,
    let's call them $\alpha$, $\beta$ and $\gamma$
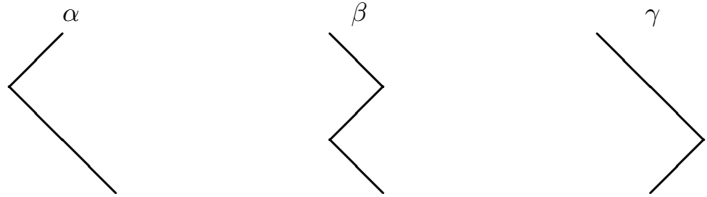
    \begin{align}
      \begin{picture}(300,80)(-30,-80)
        \thicklines
        \put(0,-10){
          \put(0,5){$\alpha$}
          \put(0,0){\line(-1,-1){20}}
          \put(-20,-20){\line(1,-1){40}}
        }
        \put(100,-10){
          \put(8,5){$\beta$}
          \put(0,0){\line(1,-1){20}}
          \put(20,-20){\line(-1,-1){20}}
          \put(0,-40){\line(1,-1){20}}
        }
        \put(200,-10){
          \put(18,5){$\gamma$}
          \put(0,0){\line(1,-1){40}}
          \put(40,-40){\line(-1,-1){20}}
        }
      \end{picture}
    \end{align}
    Hence, the $3 \times 3$ submatrices of the R-matrices $\R^{(n-1)}$ and $\R^{(n)}$ read
    \begin{align}
      \R^{(n-1)}_{\alpha\beta\gamma} = 
      \left ( \begin{array}{ccc}
        s_{11} & s_{12} & 0 \\
        s_{12} & s_{22} & 0 \\
        0 & 0 & e_L
      \end{array} \right ) \\ \notag
      \R^{(n)}_{\alpha\beta\gamma} = 
      \left ( \begin{array}{ccc}
        e_L & 0 & 0 \\
        0 & r_{11} & r_{12} \\
        0 & r_{12} & r_{22}
      \end{array} \right ),
    \end{align}
    where we've denoted by $s_{ij}$ the matrix elements of the $2 \times 2$ block from
    the level $n-1$ that we already know by the assumption of recursion,
    and by $r_{ij}$ the matrix elements of the new $2 \times 2$ block
    from level $n$ that we want to find.

    Of course, there is another possible recursion step, for which one should consider
    the following piece of the Young graph (we draw it only symbolically, but the one-hook
    diagrams standing in the vertices should be straightforward to restore)
    
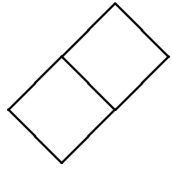
\begin{align} \label{eq:arm-step-subgraph}
      \begin{picture}(300,80)(-150,-80)
        \thicklines
        \put(0,0){\line(-1,-1){40}}
        \put(0,0){\line(1,-1){20}}
        \put(-20,-20){\line(1,-1){20}}
        \put(-40,-40){\line(1,-1){20}}
        \put(20,-20){\line(-1,-1){40}}
      \end{picture}
    \end{align}
    For which the corresponding three paths are
    
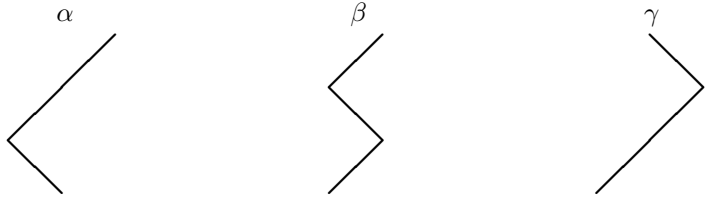
\begin{align}
      \begin{picture}(300,80)(-40,-80)
        \thicklines
        \put(0,-10){
          \put(-22,5){$\alpha$}
          \put(0,0){\line(-1,-1){40}}
          \put(-40,-40){\line(1,-1){20}}
        }
        \put(100,-10){
          \put(-12,5){$\beta$}
          \put(0,0){\line(-1,-1){20}}
          \put(-20,-20){\line(1,-1){20}}
          \put(0,-40){\line(-1,-1){20}}
        }
        \put(200,-10){
          \put(-2,5){$\gamma$}
          \put(0,0){\line(1,-1){20}}
          \put(20,-20){\line(-1,-1){40}}
        }
      \end{picture}
    \end{align}
    and the block structure of the relevent R-matrices' blocks therefore is
    \begin{align}
      \R^{(n-1)}_{\alpha\beta\gamma} = & \
      \left ( \begin{array}{ccc}
        e_A & 0 & 0 \\
        0 & s_{11} & s_{12} \\
        0 & s_{12} & s_{22}
      \end{array} \right ) \\ \notag
      \R^{(n)}_{\alpha\beta\gamma} = & \
      \left ( \begin{array}{ccc}
        r_{11} & r_{12} & 0 \\
        r_{12} & r_{22} & 0 \\
        0 & 0 & e_A
      \end{array} \right ),
    \end{align}

    Regardless of whether we use the subgraph \eqref{eq:leg-step-subgraph} or \eqref{eq:arm-step-subgraph}
    to obtain the extra equation, the following steps are very similar, so we do them explicitly
    only for \eqref{eq:arm-step-subgraph}.

    With help of the Yang-Baxter equation (again, we use only the $1,1$ matrix element of it)
    \begin{align}
      \R^{(n-1)}_{\alpha\beta\gamma} \R^{(n)}_{\alpha\beta\gamma}  \R^{(n-1)}_{\alpha\beta\gamma} =
      \R^{(n)}_{\alpha\beta\gamma} \R^{(n-1)}_{\alpha\beta\gamma}  \R^{(n)}_{\alpha\beta\gamma}
    \end{align}
    we get the closed system of equations for the matrix elements $r_{ij}$
    \begin{align} \label{eq:the-three-equations}
      e_A + e_L = r_{11} + r_{22} \\ \notag
      e_A^2 + e_L^2 = r_{11}^2 + r_{22}^2 + 2 r_{12}^2 \\ \notag
      e_A r_{11}^2 + s_{11} r_{12}^2 = e_A^2 r_{11},
    \end{align}
    from which one obtains an equation for $r_{11}$
    \begin{align}
      r_{11}^2 \left( 2 - \frac{2}{s_{11}} e_A \right) - r_{11} \left ( 2 (e_A + e_L) - \frac{2}{s_{11}} e_A^2 \right ) + 2 e_A e_L = 0
    \end{align}
    Given $e_A$ and $e_L$, this equation expresses $r_{11}$ only through $s_{11}$, that is it is a proper recursion
    for the $1,1$ matrix element of the $2 \times 2$ block.
    Moreover, the equation is \textit{homogeneous}: if one multiplies $r_{11}$, $s_{11}$, $e_A$ and $e_L$
    by same constant, the equation will not change.

    It is, therefore, an elementary check, that
    \begin{align}
      r_{11}^{(n)} = (-1)^{l+1} q^{2 a^2 + 2 a - 2 l^2 - 2 l} \cdot q^{-n (a + l + 1)} \frac{1}{[n]_\bullet}
    \end{align}
    is the solution to this equation satisfying the initial condition (the $1,1$ matrix element
    from Lemma~\ref{lem:the-very-first-doublet})
    \begin{align}
      r_{11}^{(2)} = (-1)^{l+1} q^{2 a^2 + 2 a - 2 l^2 - 2 l} \cdot q^{- 2 (a + l + 1)} \frac{1}{[2]_\bullet}
    \end{align}
    The first equation in \eqref{eq:the-three-equations} immediately gives us
    \begin{align}
      r_{22}^{(n)} = (-1)^{l} q^{2 a^2 + 2 a - 2 l^2 - 2 l} \cdot q^{-n (a + l + 1)} \frac{1}{[n]_\bullet}
    \end{align}
    and from the second equation in \eqref{eq:the-three-equations} we get
    \begin{align}
      r_{12}^{(n)} = q^{2 a^2 + 2 a - 2 l^2 - 2 l} \cdot  \frac{\sqrt{[n+1]_\bullet}}{[n]_\bullet}
    \end{align}

    Considering, in a completely analogous way, the system of equation
    corresponding to the subgraph \eqref{eq:leg-step-subgraph} we obtain the statement of the lemma,
    because we can reach any $2 \times 2$ block from the starting block \eqref{eq:picture-of-first-square}
    in a finite number of steps \eqref{eq:leg-step-subgraph} and \eqref{eq:arm-step-subgraph}.
  \end{proof}
}

{\section{Topological Schurs for 1-hook Young diagrams} \label{sec:topological-schurs}

  Let us consider the structure of the main formula \eqref{eq:homfly-definition}
  for HOMFLY (and Alexander) polynomial in the Reshetikhin-Turaev approach.

  Each summand, corresponding to a partition $\mu$ has two factors. The latter one
  is the trace of the product of $\R$-matrices in the multiplicity space.
  The former is the ratio of Schur polynomials, evaluated at the, so-called, topological locus.

  In this section we focus on the Schur polynomials' ratio.
  At $A = 1$ (i.e. for Alexander polynomial case)
  the numerator has a zero. The order of this zero is equal to the length of the main diagonal of the Young diagram $\mu$
  (see formula \eqref{eq:topological-locus-schur}).
  If the main diagonal of $\lambda$ (the partition in the denominator Schur polynomial)
  is shorter than $\mu$'s the corresponding summand vanishes.
  If they are equal, the ratio of Schur functions considerably simplifies (see below).
  Finally, if it is longer, there is a potential pole
  at $A = 1$, so the requirement of colored Alexander polynomial
  finiteness gives non-trivial relations between Vassiliev invariants
  (see \cite{paper:MST-a-new-symmetry-of-colored-alexander}
  for progress in this direction).

  In this paper we are considering only one-hook diagrams $\lambda$,
  so the third case does not arise, and the goal of this section is to prove the following
  \begin{lemma} \label{lem:ratio-of-schurs}
    Suppose $\lambda$ is a one-hook Young diagram $(a,l)$. When $\mu$ is also one-hook Young diagram $(r, s)$
    we have for the Schur polynomials' ratio at the Alexander topological locus \eqref{eq:quantum-dimension-is-schur}
    \begin{align}
      \frac{s_\mu^* (q,A)}{s_\lambda^* (q,A)} \Bigg{|}_{A=1} = (-1)^{l + s} \frac{[|\lambda|]}{[m|\lambda|]}
    \end{align}
    and when $\mu$ is more than one-hook, such ratio vanishes.
    Recall that $m$ is the number of strands and $|\mu| = m |\lambda|$.
  \end{lemma}
  \begin{proof}
    Indeed, the Schur polynomial at the topological locus is manifestly given by the product over boxes
    of its Young diagram
    \begin{align} \label{eq:topological-locus-schur}
      s_\nu^* (q,A) = \prod_{(i,j) \in \nu}
      \frac{A q^{j-i} - (A q^{j-i})^{-1}}{q^{\text{hook}(i,j)} - (q^{\text{hook}(i,j)})^{-1}},
    \end{align}
    where $\text{hook}(i,j)$ is the hook length of the box $(i,j)$.

    Let's write down explicitly the value of content $j-i$ and hook length for all boxes
    (see \eqref{eq:one-hook-partition} for the orientation of the axes)
    \begin{align}
      \begin{picture}(300,260)(0,-250)
        \put(0,0){Contents for diagram $\lambda$:}
        \put(20,-10){
          \put(0,0){\boxPic \put(7,-13){$0$}}
          \put(0,-20){\boxPic \put(4,-13){$-1$}}
          \put(0,-40){\boxPic \put(4,-13){$-2$}}
          \put(0,-80){\boxPic \put(4,-13){$-l$}}
          \put(20,0){\boxPic \put(7,-13){$1$}}
          \put(40,0){\boxPic \put(7,-13){$2$}}
          \put(80,0){\boxPic \put(7,-13){$a$}}
          \put(65,-10){$\dots$}
          \put(5,-70){$\dots$}
        }
        \put(140,0){
          \put(0,0){Contents for diagram $\mu$:}
          \put(20,-10){
            \put(0,0){\boxPic \put(7,-13){$0$}}
            \put(0,-20){\boxPic \put(4,-13){$-1$}}
            \put(0,-40){\boxPic \put(4,-13){$-2$}}
            \put(0,-80){\boxPic \put(4,-13){$-s$}}
            \put(20,0){\boxPic \put(7,-13){$1$}}
            \put(40,0){\boxPic \put(7,-13){$2$}}
            \put(80,0){\boxPic \put(7,-13){$r$}}
            \put(65,-10){$\dots$}
            \put(5,-70){$\dots$}
          }
        }
        \put(0,-140){
          \put(0,0){Hook lengths for diagram $\lambda$:}
          \put(20,-10){
            \put(0,0){\boxPic \put(3,-13){$|\lambda|$}}
            \put(0,-20){\boxPic \put(7,-13){$l$}}
            \put(0,-60){\boxPic \put(7,-13){$2$}}
            \put(0,-80){\boxPic \put(7,-13){$1$}}
            \put(20,0){\boxPic \put(7,-13){$a$}}
            \put(60,0){\boxPic \put(7,-13){$2$}}
            \put(80,0){\boxPic \put(7,-13){$1$}}
            \put(45,-10){$\dots$}
            \put(5,-50){$\dots$}
          }
          \put(140,0){
            \put(0,0){Hook lengths for diagram $\mu$:}
            \put(20,-10){
              \put(0,0){\boxPic \put(3,-13){$|\mu|$}}
              \put(0,-20){\boxPic \put(7,-13){$r$}}
              \put(0,-60){\boxPic \put(7,-13){$2$}}
              \put(0,-80){\boxPic \put(7,-13){$1$}}
              \put(20,0){\boxPic \put(7,-13){$s$}}
              \put(60,0){\boxPic \put(7,-13){$2$}}
              \put(80,0){\boxPic \put(7,-13){$1$}}
              \put(45,-10){$\dots$}
              \put(5,-50){$\dots$}
            }
          }
        }
      \end{picture}
    \end{align}
    From this it is obvious that:
    \begin{itemize}
    \item Numerator has simple zero, coming from the corner box of the diagram $\mu$
      but it cancels with simple pole, coming from the corner box of the diagram $\lambda$.
    \item Contributions of the contents of the arm of the diagram $\lambda$ cancel with the contributions
      of the hook length of the arm of the diagram $\lambda$.
      Similarly, leg contributions cancel as well, giving a factor $(-1)^l$.
      Analogously, arm and leg boxes of diagram $\mu$ give a factor $(-1)^s$.
    \item The only non-trivial contributions are coming from hook lengths of the corner boxes
      of diagrams $\lambda$ and $\mu$ and we are left with the desired expression
      (recall that $|\mu| = m |\lambda|$)
      \begin{align}
        \frac{s_\mu^* (q,A)}{s_R^* (q,A)} \Bigg{|}_{A=1} = (-1)^{l + s} \frac{[|R|]}{[m|R|]}
      \end{align}
    \end{itemize}
  \end{proof}
}
  
{\section{The proof of the main theorem} \label{sec:main-proof}
  \begin{proof}
    Now we are in position to prove the Theorem~\ref{thm:alexander-trivialization}
    -- the main theorem of the present paper.

    Consider the Alexander polynomial definition \eqref{eq:alexander-definition}.
    
    From Lemma~\ref{lem:ratio-of-schurs} we see that the ratio of quantum dimensions
    depends only on $|\lambda|$, up to a sign that can be absorbed into framing factor.
    
    Lemmas~\ref{lem:tensor-square-eigenvalues}, \ref{lem:the-very-first-doublet} and \ref{lem:arbitrary-2-by-2-block},
    together with the Lemma~\ref{lem:young-graph-structure} about the Young graph structure
    together imply, that R-matrices $\mathcal{R}^{(i)}_{\lambda,\mu}$
    (again, up to a framing factor) depend only on $|\lambda|$, moreover,
    this dependence is obtained from the one for fundamental representation by
    a simple substitution $q \rightarrow q^{|\lambda|}$.
    Hence, the trace of the arbitrary product (corresponsing to a braid)
    of the operators $\mathcal{R}^{(i)}_{\lambda,\mu}$
    is also related to the fundamental one by the substitution $q \rightarrow q^{|\lambda|}$.

    Combined, these statements prove the theorem.
  \end{proof}
}




\bibliographystyle{mpg}
\bibliography{references_alexander-factors}

\end{document}